\documentclass{scrartcl}

\usepackage{cmap}
\usepackage[utf8]{inputenc}
\usepackage[T1]{fontenc}

\usepackage{amsmath}
\usepackage{amssymb}
\usepackage{amsthm}
\usepackage{thmtools}
\usepackage{bm}
\usepackage{mathrsfs}
\usepackage{csquotes}
\usepackage{tabularx}
\usepackage{enumitem}
\usepackage{algorithm}
\usepackage{algpseudocode}
\usepackage{bbm}
\usepackage{textcomp}
\usepackage{varwidth}
\usepackage{tikz}
\usetikzlibrary{automata, calc, matrix, shapes, positioning, decorations.pathreplacing, decorations.pathmorphing, fit, arrows}

\usepackage{todonotes}

\usepackage{hyperref}

\setcapindent{0pt}

\theoremstyle{plain}
\newtheorem{thm}{Theorem}[section]
\newtheorem{lemma}[thm]{Lemma}
\newtheorem{prop}[thm]{Proposition}
\newtheorem{cor}[thm]{Corollary}

\theoremstyle{definition}

\newtheorem{counterexample}[thm]{Counter-Example}

\theoremstyle{remark}
\newtheorem{remark}[thm]{Remark}

\DeclareFontFamily{U}{mathb}{\hyphenchar\font45}
\DeclareFontShape{U}{mathb}{m}{n}{
  <5> <6> <7> <8> <9> <10> gen * mathb
  <10.95> mathb10 <12> <14.4> <17.28> <20.74> <24.88> mathb12
}{}
\DeclareSymbolFont{mathb}{U}{mathb}{m}{n}
\DeclareFontSubstitution{U}{mathb}{m}{n}
\DeclareMathSymbol{\drsh}{3}{mathb}{"EB}

\newlength{\edgelength}
\newcommand{\trans}[4]{%
  \begin{tikzpicture}[auto, shorten >=1pt, >=latex, baseline=(l.base), inner sep=0pt, outer xsep=0.3333em]
    \node (l) {\ensuremath{#1}};%
    \setlength{\edgelength}{\widthof{\scriptsize\ensuremath{#2/#3}}+0.5cm}%
    \node[base right=\edgelength of l] (r) {\ensuremath{#4}};%
    \path[->] (l.mid east) edge node[inner sep=0pt] {\scriptsize\ensuremath{#2/#3}} (r.mid west);%
  \end{tikzpicture}%
}

\newcommand{\transa}[3]{%
  \begin{tikzpicture}[auto, shorten >=1pt, >=latex, baseline=(l.base), inner sep=0pt, outer xsep=0.3333em]
    \node (l) {\ensuremath{#1}};%
    \setlength{\edgelength}{\widthof{\scriptsize\ensuremath{#2}}+0.5cm}%
    \node[base right=\edgelength of l] (r) {\ensuremath{#3}};%
    \path[->] (l.mid east) edge node[inner xsep=0pt, inner ysep=0.2em] {\scriptsize\ensuremath{#2}} (r.mid west);%
  \end{tikzpicture}%
}

\DeclareMathOperator{\id}{id}
\DeclareMathOperator{\Stab}{Stab}
\DeclareMathOperator{\Pre}{Pre}
\DeclareMathOperator{\Suf}{Suf}

\newcommand{\problem}[3][]{%
  \par\vspace{0.125cm plus 0.05cm minus 0.05cm}\begin{tabularx}{\textwidth-2\parindent}{lX}%
    \if\relax\detokenize{#1}\relax%
    \else%
      \textnormal{\textbf{Constant:}}&#1\\%
    \fi%
    \textnormal{\textbf{Input:}}&#2\\%
    \textnormal{\textbf{Question:}}&#3\\%
  \end{tabularx}\vspace{0.125cm plus 0.05cm minus 0.05cm}\par%
  }

\usepackage[affil-it]{authblk}

\author{Daniele D'Angeli\thanks{The first author was supported by the Austrian Science Fund project FWF P29355-N35.}}
\affil{Università degli Studi Niccolò Cusano\\
  Via Don Carlo Gnocchi, 3\\
  00166 Roma, Italy}
\author{Dominik~Francoeur\thanks{The second author was supported by a Doc.Mobility grant from the Swiss National Science Foundation as well as the "@raction" grant ANR-14-ACHN-0018-01 during a visit at the École Normale Supérieure in Paris.}}
\affil{Section de Math\'{e}matiques\\
  Universit\'{e} de Gen\`{e}ve\\
  2-4 rue du Lièvre\\
  1211 Genève 4, Switzerland}
\author{Emanuele~Rodaro}
\affil{Department of Mathematics\\
  Politecnico di Milano\\
  Piazza Leonardo da Vinci, 32\\
  20133 Milano, Italy}
\author{Jan~Philipp~Wächter}
\affil{Institut für Formale Methoden der Informatik (FMI)\\
  Universität Stuttgart\\
  Universitätsstraße 38\\
  70569 Stuttgart, Germany}

\title{Infinite Automaton Semigroups and Groups Have Infinite Orbits}
\begin{document}
  \maketitle

  \vspace{-2\baselineskip}

  \begin{abstract}
    We show that an automaton group or semigroup is infinite if and only if it admits an $\omega$-word (i.\,e.\ a right-infinite word) with an infinite orbit, which solves an open problem communicated to us by Ievgen V.\ Bondarenko. In fact, we prove a generalization of this result, which can be applied to show that finitely generated subgroups and subsemigroups as well as principal left ideals of automaton semigroups are infinite if and only if there is an $\omega$-word with an infinite orbit under their action. The proof also shows some interesting connections between the automaton semigroup and its dual. Finally, our result is interesting from an algorithmic perspective as it allows for a re-formulation of the finiteness problem for automaton groups and semigroups.
  \end{abstract}

  \begin{section}{Introduction}\enlargethispage{2\baselineskip}
    Automaton groups and self-similar groups became very popular after the introduction of the famous Grigorchuk group. It was the first example of a group with intermediate growth (i.e. faster than polynomial but slower than exponential), and it also has many other interesting properties. For example, it is infinite and finitely generated but each of its elements has finite order. It was also the first example of an amenable but not elementary amenable group. Soon after its introduction, it started to become clear that the most natural way to study this kind of groups is by their action on an infinite regular rooted tree, an approach which has given rise to an entirely new direction of research focusing on finitely generated groups acting by automorphisms on rooted trees and described by finite automata. Although this research revealed many interesting -- and sometimes surprising -- results about this class of so called automaton groups, the overall knowledge about them from an algebraic, algorithmic or dynamical point of view still remains limited. The dynamical view here primarily means the study of how automaton groups, which are always countable because they are finitely generated, act on the uncountable set of right infinite words, which is homeomorphic to the Cantor set. Further details can be found, for example, in \cite{MR2643891, DAngeli2019boundary, DynSubgroup, GriSa13, GrigSav16, nekrashevych2005self}. Since right infinite words can be seen as infinite paths starting at the root of an infinite regular rooted tree, this set is also often referred to as \emph{the boundary} of this tree. The action of automaton groups on the boundary seems to be very rich and best described by the structure of the corresponding Schreier graphs. In the generalized setting of semigroups instead of groups, the concept of Schreier graphs can naturally be extended into orbital graphs. One of the most natural questions in this frame is how the algebraic structure of an automaton semigroup or group influences its dynamical properties. For example, is there an infinite automaton group having only finite Schreier graphs in the boundary?
    
    \enlargethispage{2\baselineskip}
    The main result of this paper is to show that this question (communicated to us by Ievgen V.\ Bondarenko) has a negative answer; we do this in \autoref{sec:InfiniteOrbits}. In fact, our result is stronger: if one takes a subset of an automaton semigroup given by a suffix-closed language in the generators, then this subset is infinite if and only if there is an $\omega$-word with an infinite orbit under its action\footnote{Unfortunately, our proof is purely existential and does not give many information about the structure of words with infinite orbits. Some basic results towards this direction can be found in \cite{expandabilityPart}.}. This is a stronger result because, firstly, it holds in the more general setting of automaton semigroups instead of groups and, secondly, it makes a statement about certain subsets instead of only the whole automaton semigroup. These subsets include the automaton semigroup or group itself, but also finitely generated subsemigroups and subgroups -- which do not need to be automaton semigroups or groups themselves -- as well as principal left (semigroup) ideals. On the other hand, we see that we cannot generalize the result to self-similar semigroups or groups (i.\,e.\ those generated by automata with possibly infinitely many states) in \autoref{sec:SelfSimilar}.
    
    On the algorithmic side, we immediately obtain that asking whether a given automaton group or semigroup is infinite -- the so called \emph{finiteness problem} -- is equivalent to asking whether there is an $\omega$-word with an infinite orbit under the action of the automaton. While Gillibert showed that the finiteness problem for automaton semigroups is undecidable \cite{Gilbert13} (see also \cite[Theorem~4.7]{decidabilityPart} for a strengthened result), the corresponding problem for automaton groups remains an important open problem in the area \cite[7.2 b)]{GriNeShu} and, by the just mentioned connection, our result allows to consider this problem from a new perspective.
    
    The proof of our main result is heavily based on using the so-called dual automaton. We show that there is a connection between the size of certain subgraphs of Schreier graphs of an automaton group and the size of the corresponding Schreier graphs in the dual. Again, we actually show this in the more general setting of automaton semigroups (by generalizing the notion of Schreier graphs). An interesting special case are Schreier graphs for the stabilizer of a single $\omega$-word. The connection also allows to elegantly re-prove and generalize the known result that an element of an automaton group has infinite order if and only if its orbit under the action of the dual is infinite.

  \end{section}

  \begin{section}{Preliminaries}\label{sec:preliminaries}
    \paragraph{Fundamentals and Words.}
    We assume the reader to be familiar with basic notions from semigroup and group theory such as finite generation and inverses (in the group sense). We say an element $s$ of a semigroup $S$ has \emph{torsion} if there are $i, j \geq 1$ with $i \neq j$ but $s^i = s^j$; if an element does not have torsion, it is \emph{torsion-free}. This is connected to the \emph{order} of a group element $g$: it is the smallest number $i \geq 1$ such that $g^i$ is the neutral element of the group; if there is no such $i$, then the element has infinite order. Obviously, an element of a group is of infinite order if and only if it is torsion-free.
    
    To denote the disjoint union of two sets $A$ and $B$, we write $A \sqcup B$. For a \emph{partial function} from $A$ to $B$, we write $A \to_p B$. If the function is total, we omit the index $p$. 
    
    A non-empty, finite set $A$ is called an \emph{alphabet}. A finite word is a finite sequence $w = a_1 \dots a_n$ of elements $a_1, \dots, a_n \in A$; its length is $|w| = n$ and its \emph{reverse} is the mirror word\footnote{At first, the notation $\partial u$ for the reverse of $u$ might seem strange but it will make more sense after we have defined the dual of an automaton below.} $\partial w = a_n \dots a_1$. We denote the empty word by $\varepsilon$, write $A^*$ for the set of finite words over $A$ and write $A^+ = A^* \setminus \{ \varepsilon \}$. For single symbols, we also use the notation $a^*$ instead of $\{ a \}^*$ for the set $\{ a^i \mid i = 0, 1, \dots \}$. A \emph{language} $L$ is a subset of $A^*$ and its \emph{reverse} is $\partial L = \{ \partial w \mid w \in L \}$.
    
    In addition to finite words, we will also consider right-infinite sequences over $A$. Such a sequence $\alpha = a_1 a_2 \dots$ with $a_1, a_2, \dots \in A$ is called an \emph{$\omega$-word} over $A$. The set of $\omega$-words over $A$ is denoted by $A^\omega$. The term \emph{word} refers to both
     finite and $\omega$-words. Finally, we will also consider (the somewhat less standard) left-infinite sequences over $A$. However, we do not use the term word for them. Such left-infinite sequences will usually arise as the reverse $\dots a_2 a_1 = \partial \alpha$ of some $\omega$-word $\alpha = a_1 a_2 \dots$ with $a_1, a_2, \dots \in A$.
    
    A word $u$ is called a \emph{suffix} of another word $w$ if there is some finite word $x$ with $w = xu$. Symmetrically, $u$ is a \emph{prefix} of $w$ if there is a word $x$ with $w = ux$. A language $L$ is \emph{suffix-closed} if $w \in L$ implies that every suffix of $w$ is in $L$ as well; similarly, it is \emph{prefix-closed} if $w \in L$ implies $u \in L$ for all suffixes $u$ of $w$. We use $\Pre \alpha$ as the set of finite prefixes of some $\omega$-word $\alpha$. Similarly, we set $\Suf \partial \alpha = \partial \Pre \alpha$; the idea here is that $\Suf \partial \alpha$ is the set of suffixes of the left-infinite sequence $\partial \alpha$.
    
    \paragraph{Automata.}
    An $S$-automaton is a partial, finite-state, letter-to-letter transducer (without initial or final states). More formally, an \emph{$S$-automaton}\footnote{The name $S$-automaton comes from the fact that these automata generate semigroups as we will see shortly.} is a triple $\mathcal{T} = (Q, \Sigma, \delta)$ where $Q$ is an alphabet whose elements we call \emph{states}, $\Sigma$ is an alphabet as well and $\delta \subseteq Q \times \Sigma \times \Sigma \times Q$ is a set of \emph{transitions} such that for every pair $p \in Q$ and $a \in \Sigma$ the set $\{ \trans{p}{a}{b}{q} \in \delta \mid b \in \Sigma, q \in Q \}$ contains at most one element (i.\,e.\ we require the automaton to be \emph{deterministic}). Here, we have used the more graphical notation $\trans{q}{a}{b}{p}$ instead of $(q, a, b, p) \in Q \times \Sigma \times \Sigma \times Q$ for transitions. If the set $\{ \trans{p}{a}{b}{q} \in \delta \mid b \in \Sigma, q \in Q \}$ contains at least (and, thus, exactly) one element for every pair $p \in Q$ and $a \in \Sigma$, then $\mathcal{T}$ is \emph{complete}. Note that we do not require an $S$-automaton to be complete in general.
    
    When depicting an automaton graphically, we use the standard notation
    
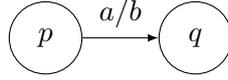
\begin{center}
      \begin{tikzpicture}[baseline=(p.base), auto, >=latex]
        \node[state] (p) {$p$};
        \node[state, right=of p] (q) {$q$};
        \draw[->] (p) edge node {$a / b$} (q);
      \end{tikzpicture}
    \end{center}
    to indicate that it contains a transition $\trans{p}{a}{b}{q}$. In addition, we also use \emph{cross diagrams}\footnote{They were originally introduced in \cite{aklmp12} -- where they are attributed to the square diagrams of \cite{glasner2005Automata} -- but seem to become more and more widespread lately.} to indicate transitions of automata. A transition $\trans{p}{a}{b}{q} \in \delta$ of some $S$-automaton $\mathcal{T} = (Q, \Sigma, \delta)$ is represented by the cross diagram
    \begin{center}
      \begin{tikzpicture}[baseline=(m-3-2.base)]
        \matrix[matrix of math nodes, text height=1.25ex, text depth=0.25ex] (m) {
          & a & \\
          p & & q \\
          & b & \\
        };
        \foreach \i in {1} {
          \draw[->] let
            \n1 = {int(2+\i)}
          in
            (m-2-\i) -> (m-2-\n1);
          \draw[->] let
            \n1 = {int(1+\i)}
          in
            (m-1-\n1) -> (m-3-\n1);
        };
      \end{tikzpicture}.
    \end{center}
    Multiple transitions can be combined into a single cross diagram.  For example, the cross diagram
    \begin{center}
      \begin{tikzpicture}
        \matrix[matrix of math nodes, text height=1.25ex, text depth=0.25ex] (m) {
                   & a_{0, 1}     &          & \dots &              & a_{0, m}     &     \\
          q_{1, 0} &              & q_{1, 1} & \dots & q_{1, m - 1} &              & q_{1, m} \\
                   & a_{1, 1}     &          &       &              & a_{1, m}     &     \\
            \vdots & \vdots       &          &       &              & \vdots       & \vdots \\
                   & a_{n - 1, 1} &          &       &              & a_{n - 1, m} &     \\
          q_{n, 0} &              & q_{n, 1} & \dots & q_{n, m - 1} &              & q_{n, m} \\
                   & a_{n, 1}     &          & \dots &              & a_{n, m}     &     \\
        };
        \foreach \j in {1, 5} {
          \foreach \i in {1, 5} {
            \draw[->] let
              \n1 = {int(2+\i)},
              \n2 = {int(1+\j)}
            in
              (m-\n2-\i) -> (m-\n2-\n1);
            \draw[->] let
              \n1 = {int(1+\i)},
              \n2 = {int(2+\j)}
            in
              (m-\j-\n1) -> (m-\n2-\n1);
          };
        };
      \end{tikzpicture}
    \end{center}
    states that the automaton contains all transitions $\trans{q_{i, j - 1}}{a_{i - 1, j}}{a_{i, j}}{q_{i, j}}$ for $1 \leq i \leq n$ and $1 \leq j \leq m$. Often, we will omit unneeded names for intermediate states or letters in cross diagrams. Since cross diagrams tend to be quite spacious, we introduce a shorthand notation. Omitting the inner states and letters, we abbreviate the above cross diagram by
    \begin{center}
      \begin{tikzpicture}
        \matrix[matrix of math nodes, text height=1.25ex, text depth=0.25ex] (m) {
                                           & u = a_{0, 1} \dots a_{0, m} & \\
          q_{n, 0} \dots q_{1, 0} = \bm{q} &                             & \bm{p} = q_{n, m} \dots q_{1, m} \\
                                           & v = a_{n, 1} \dots a_{n, m} & \\
        };
        \foreach \i in {1} {
          \draw[->] let
            \n1 = {int(2+\i)}
          in
            (m-2-\i) -> (m-2-\n1);
          \draw[->] let
            \n1 = {int(1+\i)}
          in
            (m-1-\n1) -> (m-3-\n1);
        };
      \end{tikzpicture}.
    \end{center}
    An important point to notice here is the order in which we write the states: $q_{n, 0}$ comes last but is written on the left while $q_{1, 0}$ comes first but is written on the right. Later on, it will become clearer why this is the case\footnote{Anticipating the definition below, this is because we let automaton semigroup act on the left.}.
    
    \paragraph{Automaton Semigroups.}\enlargethispage{\baselineskip}
    Let $\mathcal{T} = (Q, \Sigma, \delta)$ be an $S$-automaton. It is easy to see that for every $\bm{p} \in Q^+$ and every $u \in \Sigma^+$ there is at most one cross diagram
    \begin{center}
      \begin{tikzpicture}[baseline=(m-3-2.base)]
        \matrix[matrix of math nodes, text height=1.25ex, text depth=0.25ex] (m) {
                 & u & \\
          \bm{p} &   & \bm{q} \\
                 & v & \\
        };
        \foreach \i in {1} {
          \draw[->] let
            \n1 = {int(2+\i)}
          in
            (m-2-\i) -> (m-2-\n1);
          \draw[->] let
            \n1 = {int(1+\i)}
          in
            (m-1-\n1) -> (m-3-\n1);
        };
      \end{tikzpicture}.
    \end{center}
    If we have such a cross-diagram, we write $\bm{p} \circ u = v$ and $\bm{p} \cdot u = \bm{q}$ and extend the notation by setting $\bm{p} \circ \varepsilon = \varepsilon$, $\bm{p} \cdot \varepsilon = \bm{p}$, $\varepsilon \circ u = u$ and $\varepsilon \cdot u = \varepsilon$. Notice that, in general, we have $\bm{p} \circ u_1 u_2 = (\bm{p} \circ u_1) ((\bm{p} \cdot u_1) \circ u_2)$ for $u_1, u_2 \in \Sigma^*$ (if there is a cross-diagram with $\bm{p}$ on the left and $u_1 u_2$ at the top) and $\bm{p}_2 \circ (\bm{p}_1 \circ u) = \bm{p}_2 \bm{p}_1 \circ u$ for $\bm{p}_1, \bm{p}_2 \in Q^*$ (if there is a cross-diagram with $\bm{p}_2 \bm{p}_1$ on the left and $u$ at the top).
    
    This way, every $\bm{p} \in Q^+$ induces a partial function $\Sigma^* \to_p \Sigma^*$ mapping $u$ to $\bm{p} \circ u$. All these partial functions are length-preserving and prefix-compatible. Therefore, we can extend them into partial functions $\Sigma^* \cup \Sigma^\omega \to_p \Sigma^* \cup \Sigma^\omega$.
    
    By definition, the composition of the partial function induced by some $\bm{p}_2 \in Q^+$ with the partial function induced by some $\bm{p}_1 \in Q^+$ is the partial function induced by $\bm{p}_2 \bm{p}_1$. Therefore, the closure of the partial functions induced by the states in $Q$ under composition is exactly the set of the partial functions induces by all $\bm{p} \in Q^+$. This set forms a semigroup (with composition as operation), which we call the \emph{semigroup generated by $\mathcal{T}$} and denote by $\mathscr{S}(\mathcal{T})$. A semigroup is an \emph{automaton semigroup} if it is generated by some $S$-automaton.
    
    \begin{remark}
      If the $S$-automaton $\mathcal{T}$ is \emph{complete}, then all induced partial functions are total. In fact, in the literature, automaton semigroups are often defined using complete automata only. We will, however, consider the more general case where the automata are allowed to be non-complete as all our results hold in this setting as well. Please note that we usually state our results in the stronger of the two ways.
      
      For a discussion of how automaton semigroups generated by (partial) automata relate to automaton semigroups generated by complete automata, we refer the reader to \cite{structurePart}.
    \end{remark}
    
    \paragraph{Dual Automaton.}\enlargethispage*{3\baselineskip}
    Let $\mathcal{T} = (Q, \Sigma, \delta)$ be an $S$-automaton. We define the \emph{dual} of $\mathcal{T}$ as the $S$-automaton $\partial \mathcal{T} = (\Sigma, Q, \partial \delta)$ with the transitions
    \[
      \partial \delta = \{ \trans{a}{q}{p}{b} \mid \trans{q}{a}{b}{p} \in \delta \} \text{.}
    \]
    In other words, we exchange the roles of letters and states. The reader may verify that $\partial \mathcal{T}$ is indeed an $S$-automaton and that $\partial \mathcal{T}$ is complete if and only if $\mathcal{T}$ is.
    
    By construction, we obtain that $\mathcal{T}$ admits the cross diagram
    \begin{center}
      \begin{tikzpicture}[baseline=(m-4-1.base)]
        \matrix[matrix of math nodes, text height=1.25ex, text depth=0.25ex] (m) {
                 & a_1 &    & \dots &    & a_m &     \\
          p_1    &     & {} & \dots & {} &     & q_1 \\
                 & {}  &    &       &    & {}  &     \\[-1.5ex]
          |[inner sep=0pt]|\vdots & |[inner sep=0pt]|\vdots &    &       &    & |[inner sep=0pt]|\vdots & |[inner sep=0pt]|\vdots \\[-2.75ex]
                 & {}  &    &       &    & {}  &     \\
          p_n    &     & {} & \dots & {} &     & q_n \\
                 & b_1 &    & \dots &    & b_m &     \\
        };
        \foreach \j in {1, 5} {
          \foreach \i in {1, 5} {
            \draw[->] let
              \n1 = {int(2+\i)},
              \n2 = {int(1+\j)}
            in
              (m-\n2-\i) -> (m-\n2-\n1);
            \draw[->] let
              \n1 = {int(1+\i)},
              \n2 = {int(2+\j)}
            in
              (m-\j-\n1) -> (m-\n2-\n1);
          };
        };
      \end{tikzpicture} or, in shorthand notation,
      \begin{tikzpicture}[baseline=(m-2-3.base)]
        \matrix[matrix of math nodes, text height=1.25ex, text depth=0.25ex] (m) {
                 & u & \\
          \bm{p} &   & \bm{q} \\
                 & v & \\
        };
        \foreach \i in {1} {
          \draw[->] let
            \n1 = {int(2+\i)}
          in
            (m-2-\i) -> (m-2-\n1);
          \draw[->] let
            \n1 = {int(1+\i)}
          in
            (m-1-\n1) -> (m-3-\n1);
        };
      \end{tikzpicture}
    \end{center}
    for $\bm{p} = p_n \dots p_1$ and $\bm{q} = q_n \dots q_1$ as well as $u = a_1 \dots a_m$ and $v = b_1 \dots b_m$ if and only if $\partial \mathcal{T}$ admits the cross diagram
    \begin{center}
      \begin{tikzpicture}[baseline=(m-4-1.base)]
        \matrix[matrix of math nodes, text height=1.25ex, text depth=0.25ex] (m) {
                 & p_1 &    & \dots &    & p_n &     \\
          a_1    &     & {} & \dots & {} &     & b_1 \\
                 & {}  &    &       &    & {}  &     \\[-1.5ex]
          |[inner sep=0pt]|\vdots & |[inner sep=0pt]|\vdots &    &       &    & |[inner sep=0pt]|\vdots & |[inner sep=0pt]|\vdots \\[-2.75ex]
                 & {}  &    &       &    & {}  &     \\
          a_m    &     & {} & \dots & {} &     & b_m \\
                 & q_1 &    & \dots &    & q_n &     \\
        };
        \foreach \j in {1, 5} {
          \foreach \i in {1, 5} {
            \draw[->] let
              \n1 = {int(2+\i)},
              \n2 = {int(1+\j)}
            in
              (m-\n2-\i) -> (m-\n2-\n1);
            \draw[->] let
              \n1 = {int(1+\i)},
              \n2 = {int(2+\j)}
            in
              (m-\j-\n1) -> (m-\n2-\n1);
          };
        };
      \end{tikzpicture}
      and
      \begin{tikzpicture}[baseline=(m-2-3.base)]
        \matrix[matrix of math nodes, text height=1.25ex, text depth=0.25ex] (m) {
                     & \partial \bm{p} & \\
          \partial u &              & \partial v \\
                     & \partial \bm{q} & \\
        };
        \foreach \i in {1} {
          \draw[->] let
            \n1 = {int(2+\i)}
          in
            (m-2-\i) -> (m-2-\n1);
          \draw[->] let
            \n1 = {int(1+\i)}
          in
            (m-1-\n1) -> (m-3-\n1);
        };
      \end{tikzpicture},
      respectively;
    \end{center}
    i.\,e.\ we have to mirror the diagram along the north west to south east diagonal when passing to the dual.\newpage
    
    The dual automaton will play an important role in the following and we want to stress a connection between $\mathcal{T} = (Q, \Sigma, \delta)$ and $\partial \mathcal{T}$: we have\footnote{To avoid parenthesis, we define $\partial$ to have higher precedence than $\circ$ and $\cdot$. For example, $\partial \bm{p} \circ u$ is to be understood as $(\partial \bm{p}) \circ u$ rather than as $\partial (\bm{p} \circ u)$.} $\partial u \circ_\partial \partial \bm{p} = \partial( \bm{p} \cdot u )$ (or both undefined) for $u \in \Sigma^*$ and $\bm{p} \in Q^*$ (where the meaning of $\circ_\partial$ is the same as $\circ$ with respect to $\partial \mathcal{T}$).

    \paragraph{Union and Composition of Automata.}
    In addition to taking the dual, we need some other automaton constructions. The first construction we need is the \emph{disjoint union} $\mathcal{T} \sqcup \mathcal{T}' = (Q \sqcup Q', \Sigma \cup \Sigma', \delta \cup \delta')$ of two $S$-automata $\mathcal{T} = (Q, \Sigma, \delta)$ and $\mathcal{T}' = (Q', \Sigma', \delta')$, which is an $S$-automaton as well. Clearly, the disjoint union of two complete $S$-automata over the same alphabet is complete itself.
    
    The \emph{composition} of two $S$-automata $\mathcal{T}_1 = (Q_1, \Sigma, \delta_1)$ and $\mathcal{T}_2 = (Q_2, \Sigma, \delta_2)$ with the same alphabet $\Sigma$ is the $S$-automaton $\mathcal{T}_2 \mathcal{T}_1 = (Q_2 Q_1, \Sigma, \delta_2 \delta_1)$ where $Q_2 Q_1 = \{ q_2 q_1 \mid q_1 \in Q_1, q_2 \in Q_2 \}$ is the Cartesian product of $Q_2$ and $Q_1$ and the transitions are given by
    \[
      \delta_2 \delta_1 = \{ \trans{q_2 q_1}{a}{c}{p_2 p_1} \mid \trans{q_1}{a}{b}{p_1} \in \delta_1, \trans{q_2}{b}{c}{p_2} \in \delta_2, b \in \Sigma \} \text{.}
    \]
    The reader may verify that $\mathcal{T}_2 \mathcal{T}_1$ is indeed an $S$-automaton. If both $S$-automata are complete, then so is their composition.
    
    From the construction, it is also easy to see that the partial function induced by the state $q_2 q_1$ is indeed the composition of the partial function induced by $q_2$ with the one induced by $q_1$.
    
    The most important application of composition of automata is to take the power of some automaton. Let $\mathcal{T}^k$ denote the $k$-fold composition of $\mathcal{T} = (Q, \Sigma, \delta)$ with itself. With the above remark, we can see that the partial function induced by $\bm{p} \in Q^+$ seen as a sequence of states over $Q$ is the same as the partial function induced by $\bm{p}$ seen as a state of $\mathcal{T}^{|\bm{p}|}$. A typical application of this is the following. Suppose we have an $S$-automaton $\mathcal{T} = (Q, \Sigma, \delta)$ generating the semigroup $\mathscr{S}(\mathcal{T})$. Then, we can assume, for any sequence $\bm{q} \in Q^+$ of states, that $\bm{q}$ is already a state in $\mathcal{T}$: if this is not the case, then we can replace $\mathcal{T}$ by $\mathcal{T} \sqcup \mathcal{T}^{|\bm{q}|}$, since it generates the same semigroup.

    \paragraph{Inverse Automata and Automaton Groups.}
    An $S$-automaton $\mathcal{T} = (Q, \Sigma, \delta)$ is called \emph{invertible} if, for every pair $p \in Q$ and $b \in \Sigma$, the set $\{ \trans{p}{a}{b}{q} \in \delta \mid a \in \Sigma, q \in Q \}$ contains at most one element. For an invertible $S$-automaton $\mathcal{T} = (Q, \Sigma, \delta)$, it is possible to define its \emph{inverse} $S$-automation $\mathcal{T}^{-1} = (Q^{-1}, \Sigma, \delta^{-1})$ by
    \[
      \delta^{-1} = \{ \trans{p^{-1}}{b}{a}{q^{-1}} \mid \trans{p}{a}{b}{q} \in \delta \} \text{.}
    \]
    This way, the partial function induced by the state $q^{-1}$ of $\mathcal{T}^{-1}$ is exactly the inverse (in the sense of partial functions) of the partial function induced by the state $q$ of $\mathcal{T}$.
    
    If $\mathcal{T}$ is invertible and complete, then all sets $\{ \trans{p}{a}{b}{q} \in \delta \mid a \in \Sigma, q \in Q \}$ with $p \in Q$ and $b \in \Sigma$ contain exactly one element and all functions induced by some $p \in Q$ are bijections (and, in particular, total). Therefore, $\mathcal{T} \sqcup \mathcal{T}^{-1}$ generates a group in this case. We call this group the \emph{group generated by $\mathcal{T}$} and denote it by $\mathscr{G}(\mathcal{T})$. To emphasize this fact, we call a complete and invertible $S$-automaton a \emph{$G$-automaton}. A group is an \emph{automaton group} if it is generated by some $G$-automaton.
    
    \begin{remark}
      We will mostly work with automaton semigroups in this paper. However, the reader should not be fooled into thinking that our results are not relevant for automaton \textbf{groups}. As every automaton group is in particular an automaton semigroup, we only state our results in the more general setting.
    \end{remark}

    \paragraph{Orbital and Schreier Graphs.}
    Let $\mathcal{T} = (Q, \Sigma, \delta)$ be an $S$-automaton and let $K \subseteq Q^*$ be a set of state sequences. We define
    \[
      K \circ u = \{ \bm{q} \circ u \mid \bm{q} \in K, \bm{q} \circ u \text{ defined} \}
    \]
    as the \emph{$K$-orbit} of a word $u$ over $\Sigma$. The most important special case of this is $Q^* \circ u$, which we simply call the \emph{orbit} of $u$ (under the action defined by $\mathcal{T}$).
    
    For a language $L \subseteq \Sigma^*$, we define the relation ${\equiv_L} \subseteq Q^* \times Q^*$ by
    \[
      \bm{p} \equiv_L \bm{q} \iff \forall u \in L: \bm{p} \circ u = \bm{q} \circ u \text{ or both undefined} \text{.}
    \]
    This relation is an equivalence and we write $[ \bm{p} ]_L$ for the class of $\bm{p}$ under $\equiv_L$. We define $K/L = \{ [ \bm{q} ]_L \mid \bm{q} \in K \}$ as the set of equivalence classes of $K$. This set has a natural graph structure: we define the set of edges
    \[
      \left\{ \transa{ [ \bm{q} ]_L }{p}{ [ p\bm{q} ]_L } \mid \bm{q} \in K, p \in Q \text{ such that } p \bm{q} \in K \right\} \text{.}
    \]
    The resulting graph, thus, has out-degree bounded by $|Q|$. For simplicity, we do not distinguish between $K/L$ as a set and $K/L$ as a graph.
    
    To understand the graph $K/L$, it helps to look at some important special cases. For $L = \Sigma^*$, we have $\bm{p} \equiv_L \bm{q}$ if and only if $\bm{p} = \bm{q}$ in $\mathscr{S}(\mathcal{T})$. Therefore, the graph $Q^+ / \Sigma^*$ is (isomorphic to) the left Cayley graph of $\mathscr{S}(\mathcal{T})$. For a subset $P \subseteq Q$, $P^+ / \Sigma^*$ is the left Cayley graph of the subsemigroup generated by $P$.
    
    For the next special case, we define
    \[
      \Stab(u) = \{ \bm{q} \in Q^+ \mid u = \bm{q} \circ u \text{ and $\bm{q} \circ u$ defined} \}
    \]
    as the set of state sequences \emph{stabilizing} the word $u \in \Sigma^* \cup \Sigma^\omega$. The image of $\Stab(u)$ in $\mathscr{S}(\mathcal{T})$ is the \emph{stabilizer} of $u$. If $\mathscr{S}(\mathcal{T}) = G$ is a group, then the image of $\Stab(u)$ in $G$ is a subgroup $H$ of $G$ and we can consider the co-sets $G / H$. It is easy to see that $\bm{p} \Stab(u) =_G \bm{q} \Stab(u)$ if and only if $\bm{p} \equiv_{u} \bm{q}$ and that, thus, $Q^*/\{ u \}$ is the (left) Schreier graph of $G$ with respect to the stabilizer $H$ of $u$.
    
    Finally, let $\alpha \in \Sigma^\omega$. If $\mathcal{T}$ is complete, we have $\bm{p} \equiv_{\Pre \alpha} \bm{q}$ if and only if $\bm{p} \circ \alpha = \bm{q} \circ \alpha$. A consequence, in this case, is that the set $K / \Pre \alpha$ is in one-to-one correspondence with $K \circ \alpha$ and, similarly, that $Q^* / \Pre \alpha$ is (isomorphic to) the left orbital graph of $\alpha$ as a graph.
    
    If the automaton is non-complete, we may have $\bm{p} \not\equiv_{\Pre \alpha} \bm{q}$ although the actions of $\bm{q}$ and $\bm{p}$ are both undefined on $\alpha$. However, we still have an important connection between the two graphs: one is infinite if and only if the other one is (under certain conditions).
    \begin{lemma}\label{lem:singleWordForKOrbit}
      Let $\mathcal{T} = (Q, \Sigma, \delta)$ be an $S$-automaton and let $K \subseteq Q^*$ be suffix-closed.
      %In fact, we only need that infinitely many nodes are reachable from $[\varepsilon]_L$ in $K/L$
      Furthermore, let $L \subseteq \Sigma^*$ be a language and define
      \[
        \vec{L} = \{ \alpha \in \Sigma^\omega \mid \text{infinitely many prefixes of $\alpha$ are in $L$} \} \text{.}
      \]
      If $K/L$ is infinite, there is some $\pi = p_1 p_2 \dots \in Q^\omega$ with $p_1, p_2, \dots \in Q$ such that all $[ p_i \dots p_1 ]_L$ with $i \geq 0$ are pairwise distinct and
      \[
        \forall \alpha \in \vec{L}: \partial \Pre \pi \circ \alpha \subseteq K \circ \alpha \text{.}
      \]
    \end{lemma}
    \begin{proof}
      Assume that the graph $K/L$ is infinite. Since $K$ is suffix-closed, every node $[ q_n \dots q_1 ]_L$ with $q_n \dots q_1 \in K$ can be reached via a path from $[ \varepsilon ]_{\Sigma^*}$ whose label is $q_1 \dots q_n \in \partial K$. Therefore, infinitely many nodes are reachable from $[ \varepsilon ]_L$ and we find an infinite simple path in $K / L$ starting in $[ \varepsilon ]_L$ (because the out-degree is bounded by $|Q|$). Let $\pi = p_1 p_2 \dots$ with $p_1, p_2, \dots \in Q$ be the label of this path. Note that, in general, we do not have that every $p_i \dots p_1$ with $0 \leq i$ is in $K$! However, we have that for every $i$, there is some $\bm{q}_i \in K$ such that $[ p_i \dots p_1 ]_L = [ \bm{q}_i ]_L$. This implies that, for all $i \geq 0$ and all $u \in L$, we have $p_i \dots p_1 \circ u = \bm{q}_i \circ u \in K \circ u$. Fix some $\alpha \in \vec{L}$ and let $u_\ell$ denote the prefix of $\alpha$ of length $\ell$. We are done if we show $p_i \dots p_1 \circ u_\ell = \bm{q}_i \circ u_\ell$ for all $\ell$. However, this is the case because, for every $\ell$, there is some $\ell' \geq \ell$ with $u_{\ell'} \in L$ since $\alpha$ is in $\vec{L}$.
    \end{proof}
    
    \begin{prop}\label{prop:KSchreierAndKOrbit}
      Let $\mathcal{T} = (Q, \Sigma, \delta)$ be an $S$-automaton and let $K \subseteq Q^*$ be suffix-closed.
      %Again: we only need that infinitely many nodes are reachable from $[\varepsilon]_L$ in $K/L$
      For all $\alpha \in \Sigma^\omega$, we have
      \[
        | K / \Pre \alpha | = \infty \iff | K \circ \alpha | = \infty \text{.}
      \]
    \end{prop}
    \begin{proof}
      If $K \circ \alpha$ is infinite, there are infinitely many $\bm{q}_0, \bm{q}_1, \dots \in K$ such that all $\bm{q}_i \circ \alpha$ are pairwise different. In particular, they must be defined on all prefixes of $\alpha$ and, for every $i, j \geq 0$ with $i \neq j$, there must be some finite prefix $u$ of $\alpha$ such that $\bm{q}_i \circ u$ is different to $\bm{p}_j \circ u$. This shows that $K / \Pre \alpha$ is infinite.
      
      Now assume that $K / \Pre \alpha$ is infinite. Let $\pi = p_1 p_2 \dots \in Q^\omega$ with $p_1, p_2, \dots \in Q$ denote the one from \autoref{lem:singleWordForKOrbit} (for $L = \Pre \alpha$). Thus, we have $\partial \Pre \pi \circ \alpha \subseteq K \circ \alpha$ and we only need to show that the former is infinite. If there is some $i_0 \geq 0$ such that $p_{i_0} \dots p_1 \circ \alpha$ is not defined, then there is some shortest finite prefix $u$ of $\alpha$ for which $p_{i_0} \dots p_1 \circ u$ is not defined. Notice that no $p_j \dots p_1 \circ u$ with $j \geq i_0$ can be defined. This implies that the infinitely many $p_j \dots p_1$ with $j \geq i_0$ must all act pairwise differently on those prefixes of $\alpha$ that are shorter than $u$. This, however, is not possible since there are only finitely many possible (partial) actions. Therefore, all $p_i \dots p_1 \circ \alpha$ with $i \geq 0$ must be defined but they are all pairwise different (since there must be a difference on some prefix of $\alpha$ as we have $[p_i \dots p_1]_{\Pre \alpha} \neq [p_j \dots p_1]_{\Pre \alpha}$ for all $i \neq j$).
    \end{proof}
    
    An important special case of \autoref{prop:KSchreierAndKOrbit}, which we will be using a few times below, is if $K$ belongs to a single $\omega$-word.
    \begin{cor}
      \label{lem:SchreierInfiniteIffOrbitInfinite}
      Let $\mathcal{T} = (Q, \Sigma, \delta)$ be an $S$-automaton and let $\pi \in Q^\omega$ and $\alpha \in \Sigma^\omega$. Then, we have:
      \[
        | \partial \Pre \pi / \Pre \alpha | = \infty \iff | \partial \Pre \pi \circ \alpha | = \infty
      \]
    \end{cor}
    
    \paragraph{Dual Orbits.}
    Let $\mathcal{T} = (Q, \Sigma, \delta)$ be an $S$-automaton. We can also see $L \subseteq \Sigma^*$ as a set of state sequences and $K \subseteq Q^*$ as a language over the alphabet of the dual automaton $\partial \mathcal{T}$. Under this view, we obtain a well-defined meaning of $L/K$ from the definition above. So, $L/K$ is to be understood with respect to the dual $\partial \mathcal{T}$ of $\mathcal{T}$ while $K/L$ is to be understood with respect to $\mathcal{T}$.
    
    Similarly, we can also consider $L \circ_\partial \bm{p}$ the $L$-orbit of some word $\bm{p}$ over $Q$ under the action of $\partial \mathcal{T}$. If $\bm{p}$ is a finite word, then we have a one-to-one correspondence between $L \circ_\partial \bm{p}$ and
    \[
      \partial \bm{p} \cdot \partial L = \{ \partial \bm{p} \cdot \partial u \mid u \in L, \partial \bm{p} \cdot \partial u \text{ defined} \}
    \]
    where the bijection is given by taking the reverse. With this in mind, we can naturally extend the notation $\bm{p} \cdot u$ and $\bm{p} \cdot L$ to left-infinite sequences over $Q$: $\partial \pi \cdot u$ is the reverse of $\partial u \circ_\partial \pi$ for $\pi \in Q^\omega$ and $\partial \pi \cdot L$ is the set of all $\partial \pi \cdot u$ for $u \in L$.
  \end{section}

  \begin{section}{Infinite Automaton Semigroups Have Infinite Orbits}\label{sec:InfiniteOrbits}
    In this section, we are going to show that an automaton semigroup $\mathscr{S}(\mathcal{T})$ generated by an $S$-automaton $\mathcal{T} = (Q, \Sigma, \delta)$ is infinite if and only if there is some $\omega$-word $\alpha \in \Sigma^\omega$ with an infinite orbit $Q^* \circ \alpha$. In fact, we are going to show a more general result, which we apply also for finitely generated infinite subsemigroups of $\mathscr{S}(\mathcal{T})$ and for principal left ideals.
    
    For our discussion, fix some arbitrary $S$-automaton $\mathcal{T} = (Q, \Sigma, \delta)$ for this section. It is well-known that $\mathcal{T}$ generates an infinite semigroup if and only if $\partial \mathcal{T}$ generates an infinite semigroup as well (see e.\,g.\ \cite[Proposition~10]{aklmp12} or \cite[Corollary~1]{DaRo14}). We will see that this connection between the semigroup and its dual is only a special case of a more fundamental one:
    \begin{prop}
      \label{prop:KLInfiniteIffLKInfinite}
      Let $K \subseteq Q^*$ be suffix-closed and let $L \subseteq \Sigma^*$ be prefix-closed. Then, we have:
      \[
        | K / L | = \infty \iff | \partial L / \partial K | = \infty
      \]
    \end{prop}
    \begin{proof}
      Due to duality, we only have to show one direction, which we do by contraposition. Assume that $K / L = \{ [\bm{p}]_L \mid \bm{p} \in K \}$ is of finite size $C$. Clearly, the size of $K \circ u$ for some $u \in L$ is bounded by the number $C$ of different classes. For $u \in L$ and $\bm{p} \in K$, we have the following equivalence of cross diagrams:
      \begin{center}
        \begin{tikzpicture}[baseline=(m-2-3.base)]
          \matrix[matrix of math nodes, text height=1.25ex, text depth=0.25ex] (m) {
                   & u & \\
            \bm{p} &   & \bm{q} \\
                   & v & \\
          };
          \draw[->] (m-2-1) -> (m-2-3);
          \draw[->] (m-1-2) -> (m-3-2);
          \node[below=0cm of m-3-2] {in $\mathcal{T}$};
          \node[base right=-0.3333em of m-1-2, inner sep=0pt] {${}\in L$};
          \node[base left=-0.3333em of m-2-1, inner sep=0pt] {$K \ni{}$};
          \node[base right=-0.3333em of m-3-2, inner sep=0pt] {${}\in K \circ u$};
        \end{tikzpicture}
        ${}\iff{}$
        \begin{tikzpicture}[baseline=(m-2-3.base)]
          \matrix[matrix of math nodes, text height=1.25ex, text depth=0.25ex] (m) {
                       & \partial \bm{p} & \\
            \partial u &                 & \partial v \\
                       & \partial \bm{q} & \\
          };
          \draw[->] (m-2-1) -> (m-2-3);
          \draw[->] (m-1-2) -> (m-3-2);
          \node[below=0cm of m-3-2] {in $\partial \mathcal{T}$};
          \node[base right=-0.3333em of m-1-2, inner sep=0pt] {${}\in \partial K$};
          \node[base right=-0.3333em of m-2-3, inner sep=0pt] {${}\in \partial (K \circ u)$};
        \end{tikzpicture}
      \end{center}
      It follows that the action of $\partial u \in \partial L$ on a (dual) word $\partial \bm{p} \in \partial K$ is described by an $S$-automaton\footnote{In fact, the $S$-automaton is basically the part of $(\partial \mathcal{T})^{|\partial u|}$ reachable from $\partial u$ by input words from $\partial K$. From the cross diagram on the right, we know that all such reachable states are from $\partial (K \circ u)$.} with state set $\partial (K \circ u)$. There are only finitely many (non-isomorphic) $S$-automata of size $|\partial (K \circ u)| = |K \circ u| \leq C$. If we have $[ \partial u ]_{\partial K} \neq [ \partial u' ]_{\partial K}$, then the automaton belonging to $\partial u$ must be different to the one belonging to $\partial u'$. Therefore, we obtain that there are only finitely many different classes or, in other words, that $\partial L / \partial K$ is finite.
    \end{proof}
        
    Our main result follows from \autoref{prop:KLInfiniteIffLKInfinite}. However, \autoref{prop:KLInfiniteIffLKInfinite} is also interesting on its own! To demonstate this, we will present some further applications at the end of this section. First, however, we show our main result and its most interesting corollaries.
    
    \begin{thm}
      \label{thm:infiniteSubsetHasInfiniteOrbit}
      Let $K \subseteq Q^*$ be suffix-closed. The image of $K$ in $\mathscr{S}(\mathcal{T})$ is infinite if and only if there is some $\omega$-word $\alpha$ whose $K$-orbit $K \circ \alpha$ is infinite.
    \end{thm}
    \begin{proof}
      If the image of $K$ in $\mathscr{S}(\mathcal{T})$ is finite, then $K \circ \alpha$ is clearly bounded by that finite size for all $\alpha \in \Sigma^\omega$.

%      Old version not using the lemma above
%      Therefore, assume that $K$ is infinite in $\mathscr{S}(\mathcal{T})$. Since $\equiv_{\Sigma^*}$ is the equality in $\mathscr{S}(\mathcal{T})$, we have that the graph $K / \Sigma^*$ is infinite and, since $K$ is suffix-closed, every node $[ q_n \dots q_1 ]_{\Sigma^*}$ with $q_n \dots q_1 \in K$ can be reached via a path from $[ \varepsilon ]_{\Sigma^*}$ whose label is $q_1 \dots q_n \in \partial K$. Therefore, infinitely many nodes are reachable from $[ \varepsilon ]_{\Sigma^*}$ and we find an infinite simple path in $K / \Sigma^*$ starting in $[ \varepsilon ]_{\Sigma^*}$ (because the out-degree is bounded). Let $\pi = p_1 p_2 \dots$ with $p_1, p_2, \dots \in Q$ be the label of this path. Note that, in general, we do not have that every $p_i \dots p_1$ with $0 \leq i$ is in $K$! However, we have that for every $i$, there is some $\bm{q}_i \in K$ such that $[ p_i \dots p_1 ]_{\Sigma^*} = [ \bm{q}_i ]_{\Sigma^*}$ or, equivalently, $\bm{q}_i = p_i \dots p_1$ in $\mathscr{S}(\mathcal{T})$. Therefore, we have $\partial \Pre \pi \circ \alpha \subseteq K \circ \alpha$ for all $\alpha \in \Sigma^\omega$ and it suffices to show that there is some $\alpha \in \Sigma^\omega$ with $| \partial \Pre \pi \circ \alpha | = \infty$.

      Therefore, assume that $K$ is infinite in $\mathscr{S}(\mathcal{T})$. Since $\equiv_{\Sigma^*}$ is the equality in $\mathscr{S}(\mathcal{T})$, we have that the graph $K / \Sigma^*$ is infinite. From \autoref{lem:singleWordForKOrbit}, we obtain that there is some $\pi = p_1 p_2 \dots$ with $p_1, p_2, \dots \in Q$ with $\partial \Pre \pi \circ \alpha \subseteq K \circ \alpha$ for all $\alpha \in \Sigma^\omega$. Thus, it suffices to show that there is some $\alpha \in \Sigma^\omega$ with $| \partial \Pre \pi \circ \alpha | = \infty$.
      
      We have that $\partial \Pre \pi / \Sigma^*$ is infinite since it contains the infinite simple path starting in $[ \varepsilon ]_{\Sigma^*}$ whose label is $\pi$. By \autoref{prop:KLInfiniteIffLKInfinite}, we obtain that $\Sigma^* / \Pre \pi$ is infinite as well and, thus, contains an infinite simple path starting in $[ \varepsilon ]_{\Pre \pi}$. Let $\alpha \in \Sigma^\omega$ be the label of this path. Obviously, this infinite simple path exists also in the subgraph $\partial \Pre \alpha / \Pre \pi$, which shows that this graph is also infinite. By applying \autoref{prop:KLInfiniteIffLKInfinite} again, we obtain that $\partial \Pre \pi / \Pre \alpha$ is infinite. It then follows from \autoref{lem:SchreierInfiniteIffOrbitInfinite} that $\partial \Pre \pi \circ \alpha$ is infinite.
    \end{proof}
    
    \paragraph{Corollaries.}
    The formulation of \autoref{thm:infiniteSubsetHasInfiniteOrbit} is quite general and allows us to derive a few natural corollaries. The first one is the special case of \autoref{thm:infiniteSubsetHasInfiniteOrbit} where we set $L = Q^*$.
    
    \begin{cor}\label{cor:infiniteSemigroupsHaveInfiniteOrbits}
      The semigroup $\mathscr{S}(\mathcal{T})$ generated by some $S$-automaton $\mathcal{T} = (Q, \Sigma, \delta)$ is infinite if and only if there is some $\omega$-word $\alpha \in \Sigma^\omega$ with an infinite orbit $Q^* \circ \alpha$.
    \end{cor}
    
    Using Gillbert's result on the undecidability of the finiteness problem for automaton semigroup \cite{Gilbert13}, we immediately also obtain the following consequence.
    \begin{cor}
      The problem
      \problem{
        an $S$-automaton $\mathcal{T} = (Q, \Sigma, \delta)$
      }{
        is there an $\omega$-word $\alpha \in \Sigma^\omega$ with $|Q^* \circ \alpha| = \infty$?
      }
      \noindent{}is undecidable.
    \end{cor}
    
    \enlargethispage{\baselineskip}
    Since every automaton group is, in particular, an automaton semigroup, we immediately have \autoref{cor:infiniteSemigroupsHaveInfiniteOrbits} also for automaton groups, which negatively answers an open question communicated to us by Ievgen V. Bondarenko (see also \cite[Open Problem~4.3]{decidabilityPart}): is there an infinite automaton group having only finite Schreier graphs in the boundary?
    \begin{cor}\label{cor:infiniteGroupsHaveInfiniteSemiOrbits}
      Let $\mathcal{T} = (Q, \Sigma, \delta)$ be a $G$-automaton.
      
      Then, $\mathscr{G}(\mathcal{T})$ is infinite if and only if there is an $\omega$-word $\alpha \in \Sigma^\omega$ such that $Q^* \circ \alpha \subseteq {Q^{\pm 1}}^* \circ \alpha$ is infinite
    \end{cor}
    
    This result is also interesting algorithmically as it allows for a re-formulation of the finiteness problem for automaton groups.
    \begin{cor}
      The finiteness problem for automaton groups
      \problem{
        a $G$-automaton $\mathcal{T}$
      }{
        is $\mathscr{G}(\mathcal{T})$ infinite?
      }
      \noindent{}is equivalent to the problem
      \problem{
        a $G$-automaton $\mathcal{T} = (Q, \Sigma, \delta)$
      }{
        is there an $\omega$-word $\alpha \in \Sigma^\omega$ with $|Q^* \circ \alpha| = \infty$?
      }
    \end{cor}

    Next, we can extend our result to finitely generated subsemigroups: a finitely generated subsemigroup of an automaton semigroup is infinite if and only if it admits an $\omega$-word which has an infinite orbit under its action. This is even true if the subsemigroup itself is not an automaton semigroup!
    \begin{cor}\label{cor:infiniteSubsemigroupsHaveInfiniteOrbist}
      Let $\mathcal{T} = (Q, \Sigma, \delta)$ be an $S$-automaton. A subsemigroup of $\mathscr{S}(\mathcal{T})$ generated by a finite set $\bm{P} \subseteq Q^+$ is infinite if and only if there is an $\omega$-word $\alpha \in \Sigma^\omega$ whose orbit $\bm{P}^* \circ \alpha$ under the action of the subsemigroup is infinite.
    \end{cor}
    \begin{proof}
      By replacing $\mathcal{T}$ by the union of $\mathcal{T}$ with appropriate powers of itself, we may assume that all $\bm{p} \in \bm{P}$ are actually states in $\mathcal{T}$. Then, $\bm{P}^*$ is a suffix-closed language over the states of $\mathcal{T}$ and the result follows from \autoref{thm:infiniteSubsetHasInfiniteOrbit}.
    \end{proof}
    
    An interesting direct application of \autoref{cor:infiniteSubsemigroupsHaveInfiniteOrbist} is that an element of an automaton semigroup is torsion-free if and only if there is a single $\omega$-word on which all powers of the element act differently.
    \begin{cor}\label{cor:torsionFreenessIsWitnessedByASingleWord}
      Let $\mathcal{T} = (Q, \Sigma, \delta)$ be an $S$-automaton and let $\bm{q} \in Q^+$ be some state sequence. Then, $\bm{q}$ is torsion-free in $\mathscr{S}(\mathcal{T})$ if and only if there is some $\omega$-word $\alpha \in \Sigma^\omega$ such that $\bm{q}^i \circ \alpha \neq \bm{q}^j \circ \alpha$ for all $i \neq j$.
      
      If $\mathcal{T}$ is a $G$-automaton, then $\bm{q}$ has infinite order in $\mathscr{G}(\mathcal{T})$ if and only if there is some $\omega$-word $\alpha \in \Sigma^\omega$ such that $\bm{q}^i \circ \alpha \neq \alpha$ for all $i > 0$.
    \end{cor}
    
    We also get an analogous result for principal left ideals of automaton semigroups. A subset $I$ of some semigroup $S$ is a \emph{left ideal} if $SI \subseteq I$ holds. A \emph{principal} left ideal is a left ideal of the form $S^1 s$ for an $s \in S$ where $S^1$ is the monoid generated by $S$. We get that a principal left ideal of an automaton semigroup is infinite if and only if it admits an $\omega$-word with an infinite orbit. In this context, (principal) left ideals are interesting since they are not finitely generated in general (although they are always subsemigroups). Therefore, this result does not immediately follow from \autoref{cor:infiniteSubsemigroupsHaveInfiniteOrbist}.
    \begin{cor}\label{cor:infiniteLeftPrincipalIdealsHaveInfiniteOrbits}
      Let $\mathcal{T} = (Q, \Sigma, \delta)$ be an $S$-automaton and let $\bm{p} \in Q^+$. Then, the principal left ideal $Q^* \bm{p} = \{ \bm{q} \bm{p} \mid \bm{q} \in Q^* \}$ in $\mathscr{S}(\mathcal{T})$ is infinite if and only if there is an $\omega$-word $\alpha \in \Sigma^\omega$ whose orbit $Q^* \bm{p} \circ \alpha$ under the action of the left ideal is infinite.
    \end{cor}
    \begin{proof}
      Again, we can consider $\bm{p}$ as a state in $\mathcal{T}^{|\bm{p}|}$ and replace $\mathcal{T}$ by $\mathcal{T} \cup \mathcal{T}^{|\bm{p}|}$ without changing the semigroup. This turns $Q^* \bm{p} \cup \{ \varepsilon \}$ into a suffix-closed language and the result follows from \autoref{thm:infiniteSubsetHasInfiniteOrbit}.
    \end{proof}
    
    However, not every interesting subsemigroup of an automaton semigroup belongs to a language of generators which is suffix-closed. One such example are principal right ideals: a subset $I$ of a semigroup $S$ is a \emph{right ideal} if $IS \subseteq I$ holds; a \emph{principal} right ideal is a right ideal of the form $s S^1$ for some $s \in S$. If $I$ is both a left and a right ideal, it is called a \emph{two-sided ideal}. \emph{Principal} two-sided ideals are two-sided ideals of the form $S^1 s S^1$ for some $s \in S$. The analogous result to \autoref{cor:infiniteLeftPrincipalIdealsHaveInfiniteOrbits} for principal right and two-sided ideals does not hold, as the next counter-example shows. Since all three types of ideals are also subsemigroups, this shows as well that \autoref{cor:infiniteSubsemigroupsHaveInfiniteOrbist} cannot be generalized to arbitrary (non-finitely generated) subsemigroups.
    % What about complete automaton semigroups?
    \begin{counterexample}
      A counter-example is given by the $S$-automaton
      
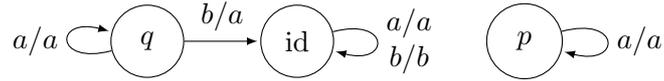
\begin{center}
        \begin{tikzpicture}[auto, shorten >=1pt, >=latex, baseline=(p.base)]
          \node[state] (q) {$q$};
          \node[state, right=of q] (id) {$\id$};\texttt{
          \node[state, right=2cm of id] (p) {$p$};}
          \path[->] (q) edge[loop left] node {$a/a$} (q)
                        edge node {$b/a$} (id)
                    (id) edge[loop right] node[align=center] {$a/a$\\$b/b$} (id)
                    (p) edge[loop right] node {$a/a$} (p)
          ;
        \end{tikzpicture}
      \end{center}
      whose state set we denote by $Q$. For $L = pQ^*$, we then have that, for every $\alpha \in \Sigma^\omega$, either $L \circ \alpha = \emptyset$ or $L \circ \alpha = \{ a^\omega \}$ contains only a single word. However, the principal right ideal (and non-finitely generated subsemigroup) given by $L= pQ^*$ in $\mathscr{S}(\mathcal{T})$ is infinite: if we take $pq^i \in L$ and $pq^j \in L$ for $i > j$, then $pq^i \circ b^i = a^i$ while $pq^j \circ{} b^i$ is undefined (see the schematic depiction of the orbital graph of $b^i$ in \autoref{fig:orbitalGraphOfBI}); thus, we have found infinitely many elements in $L$ that are pairwise distinct in $\mathscr{S}(\mathcal{T})$. Finally, we have the same situation for the principal two-sided ideal $Q^* p Q^*$ in $\mathscr{S}(\mathcal{T})$ since we have $qp = \id p = p$ in $\mathscr{S}(\mathcal{T})$.
    \end{counterexample}
    \begin{figure}[t]%
      \centering%
        \begin{tikzpicture}[auto, shorten >=1pt, >=latex]
          \node (bi) {$b^i$};
          \node[right=of bi] (abi-1) {$a b^{i - 1}$};
          \node[right=of abi-1] (dots) {$\dots$};
          \node[right=of dots] (ai-1b) {$a^{i - 1} b$};
          \node[right=of ai-1b] (ai) {$a^i$};
          
          \path[->] (bi) edge node {$q$} (abi-1)
                    (abi-1) edge node {$q$} (dots)
                    (dots) edge node {$q$} (ai-1b)
                    (ai-1b) edge node {$q$} (ai)
                    (ai) edge[loop right] node {$p, q$} (ai)
          ;
        \end{tikzpicture}%
      \caption{Orbital graph of $b^i$ with $\id$-self-loops omitted for clarity.}\label{fig:orbitalGraphOfBI}
    \end{figure}
    
    \paragraph{Corollaries of \autoref{prop:KLInfiniteIffLKInfinite}.}
    We finally return to other applications of \autoref{prop:KLInfiniteIffLKInfinite}. A particularly nice special case is the restriction to two single $\omega$-words.
    \begin{cor}
      Let $\pi \in Q^\omega$ and $\alpha \in \Sigma^\omega$. Then, we have
      \[
        |\partial \Pre \pi \circ \alpha| = \infty \iff |\partial \Pre \alpha \circ_\partial \pi| = |\partial \pi \cdot \Pre \alpha| = \infty
      \]
    \end{cor}
    \begin{proof}
      We only have to show one direction as the other one follows from duality. Assume that $\partial \Pre \pi \circ \alpha$ is infinite. By \autoref{lem:SchreierInfiniteIffOrbitInfinite}, $\partial \Pre \pi / \Pre \alpha$ is infinite and, by \autoref{prop:KLInfiniteIffLKInfinite}, we obtain that $\partial \Pre \alpha / \Pre \pi$ is infinite as well. This means that there are infinitely many prefixes of $\alpha$ that (dually) act pairwise differently on $\pi$. Therefore, $\partial \Pre \alpha \circ_\partial \pi$ is infinite.
    \end{proof}
    
    Additionally, we can recover a known connection between the order of a group element and its orbit under the action of the dual (see \cite[Theorem~3]{DaRo16}). In fact, we can easily generalize this connection to semigroups and ultimately periodic words.\footnote{The connection seems to be quite versatile. For example, it can also be generalized in different ways (see, e.\,g.\ \cite[Lemma~3.6]{francoeur2018existence}).}
    \begin{thm}\label{thm:torsionIsFiniteDualOrbit}
      Let $\mathcal{T} = (Q, \Sigma, \delta)$ be an $S$-automaton and let $\bm{q} = q_1 q_2 \dots q_n$ be a non-empty sequence of states $q_1, q_2, \dots, q_n \in Q$. Then, the statements
      \begin{enumerate}
        \item $\partial \bm{q}$ has torsion in $\mathscr{S}(\mathcal{T})$.
        \item The orbit $\Sigma^* \circ_{\partial} \bm{q}^\omega$ of $\bm{q}^\omega$ under the action of the dual of $\mathcal{T}$ is finite.
        \item The orbit $\Sigma^* \circ_{\partial} \bm{p} \bm{q}^\omega$ of $\bm{p} \bm{q}^\omega$ under the action of the dual of $\mathcal{T}$ is finite for all $\bm{p} \in Q^*$.
      \end{enumerate}
      are equivalent.
    \end{thm}
    \begin{proof}
      The implication from 3.~to 2.~is trivial and, for the other direction, we observe that we have $\Sigma^* \circ_\partial \bm{p} \bm{q}^\omega \subseteq Q^{|\bm{p}|} (\Sigma^* \circ_\partial \bm{q}^\omega)$ for all $\bm{p} \in Q^*$.
      
      The equivalence of 1.~and 2.~can be shown directly from the results above. We set $q' = \partial \bm{q} = q_n \dots q_1 \in Q^n$ and consider it as a single state in $\mathcal{T}^n$. If we define ${}^{\omega} q'$ as the left infinite sequence $\dots q' q'$ over $Q^n$, we have
      \begin{align*}
        | \Sigma^* \circ_{\partial \mathcal{T}} \bm{q}^\omega | < \infty &\iff | \partial (\bm{q}^\omega) \cdot_{\mathcal{T}} \Sigma^* | < \infty &&\text{(definition of dual orbits)}\\
        &\iff | {}^{\omega} q' \cdot_{\mathcal{T}^n} \Sigma^* | < \infty &&\text{(definition of power automata)}\\
        &\iff | \Sigma^* \circ_{\partial(\mathcal{T}^n)} (q')^\omega | < \infty &&\text{(definition of dual orbits)}\\
        &\iff | \Sigma^* / \Pre (q')^\omega | < \infty \text{ w.\,r.\,t.~} \partial(\mathcal{T}^n) &&\text{(\autoref{prop:KSchreierAndKOrbit})}\\
        &\iff | \partial \Pre (q')^\omega / \Sigma^* | < \infty \text{ w.\,r.\,t.~} \mathcal{T}^n &&\text{(\autoref{prop:KLInfiniteIffLKInfinite})}\\
        &\iff q' \text{ has torsion in } \mathscr{S}(\mathcal{T}^n) &&\text{($\equiv_{\Sigma^*}$ is equality in $\mathscr{S}(\mathcal{T}^n)$)}\\
        &\iff \partial \bm{q} \text{ has torsion in } \mathscr{S}(\mathcal{T}) &&\text{(definition of power automata)}
      \end{align*}
      where we use index notation to indicate the automaton to which $\circ$ and $\cdot$ refer.
    \end{proof}
  \end{section}
  
  \begin{section}{Self-Similar Semigroups}\label{sec:SelfSimilar}
    So far, we considered all automata to be finite, in the sense that both the state set and the alphabet were finite. For this section, we are going to relax this constraint on the state set. We will call semigroups generated by such infinite state $S$-automata \emph{self-similar}.

    It turns out that, in the setting of self-similar semigroups and groups, it is rather easy to construct a counter example for the analogue of \autoref{cor:infiniteSemigroupsHaveInfiniteOrbits}: the automaton given in \autoref{fig:InfiniteButFiniteOrbits} generates an infinite self-similar semigroup (and also an infinite self-similar group), but all its orbits are finite.
    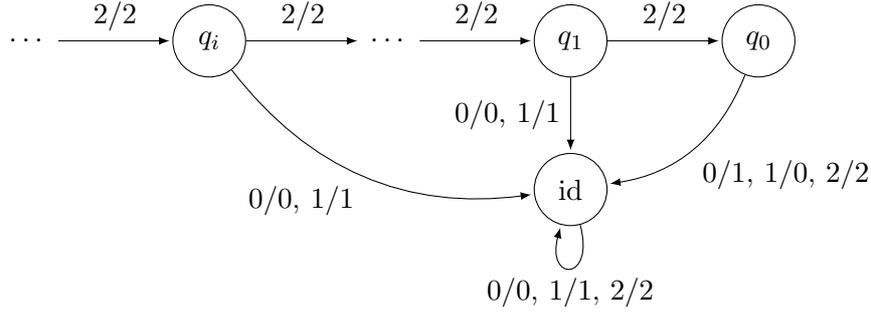
\begin{figure}[h]
      \begin{center}
        \begin{tikzpicture}[auto, shorten >=1pt, >=latex, node distance=1cm and 1.5cm]
          \node (dots1) {$\cdots$};
          \node[state, right=of dots1] (an) {$q_i$};
          \node[right=of an] (dots2) {$\cdots$};
          \node[state, right=of dots2] (a2) {$q_1$};
          \node[state, right=of a2] (a1) {$q_0$};
          \node[state, below=of a2] (id) {$\id$};
          
          \draw[->] (dots1) edge node {$2/2$} (an)
                    (an) edge node {$2/2$} (dots2)
                         edge[bend right] node[swap] {$0/0$, $1/1$} (id)
                    (dots2) edge node {$2/2$} (a2)
                    (a2) edge node {$2/2$} (a1)
                         edge node[swap] {$0/0$, $1/1$} (id)
                    (a1) edge[bend left] node {$0/1$, $1/0$, $2/2$} (id)
                    (id) edge[loop below] node {$0/0$, $1/1$, $2/2$} (id);
        \end{tikzpicture}%
      \end{center}%
      \caption{An infinite state $G$-automaton generating an infinite semigroup (and group) such that every orbit is finite.}\label{fig:InfiniteButFiniteOrbits}
    \end{figure}

    However, in the example of \autoref{fig:InfiniteButFiniteOrbits}, the orbits are uniformly bounded; in fact, they have size at most $2$. It is thus natural to ask if, for a self-similar semigroup, either the orbits are uniformly bounded by a constant or there exists an infinite orbit. This turns out to be false, as can be seen in the following example, suggested to us by Laurent Bartholdi.
      
    \begin{counterexample}\label{counterex:unboundedButOnlyFiniteOrbits}
      Let $\Sigma = \{0, 1, 2\}$ and $Q = \{\id, q_{ij} \mid 0 < i,\, 1 \leq j \leq i^2 \}$ and define the map $\tau: Q \times \Sigma \to \Sigma \times Q$ by $\tau(\id, a) = (a, \id)$ for all $a \in \Sigma$ and
      \begin{align*}
        \tau(q_{ij}, 0) &=
          \begin{cases}
            (1, q_{i(j-1)}) & \text{ if } j\equiv 1 \mod i\\
            (0, \id) & \text{ otherwise}
          \end{cases} \\
        \tau(q_{ij}, 1) &=
          \begin{cases}
            (0, q_{i(j-1)}) & \text{ if } j\equiv 1 \mod i\\
            (1, \id) & \text{ otherwise}
          \end{cases} \\
        \tau(q_{ij},2) &=
          \begin{cases}
          (2, \id) & \text{ if } j\equiv 1 \mod i\\
          (2, q_{i(j-1)}) & \text{ otherwise}
        \end{cases}
      \end{align*}
      (where we set $q_{i0}=\id$). This induces a $G$-automaton $\mathcal{T} = (Q, \Sigma, \delta)$ by
      \[
        \delta = \{ \trans{q}{a}{b}{p} \mid q, p \in Q, a, b \in \Sigma, \tau(q, a) = (b, p) \}
      \]
      A part of the automaton is represented in \autoref{figure:NonBoundedExample}. It is easy to see that the semigroup generated by $\mathcal{T}$ coincides with the group generated by $\mathcal{T}$.
      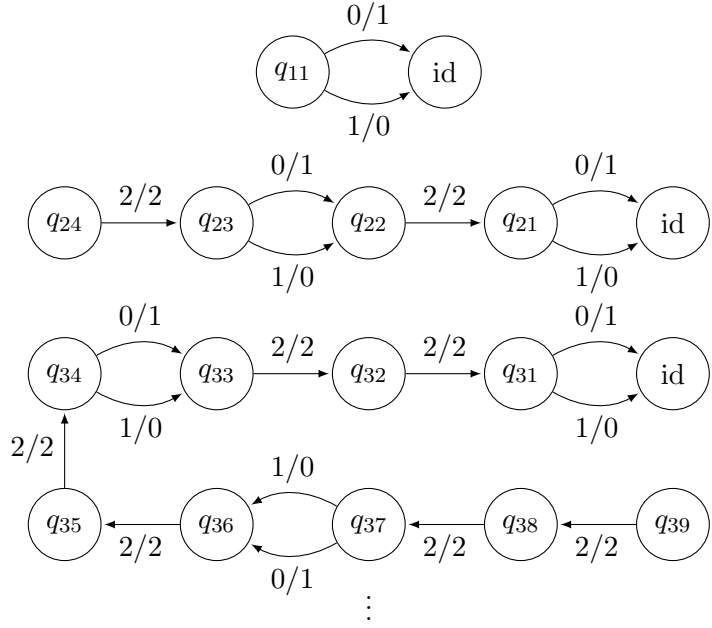
\begin{figure}[h]
        \begin{center}
          \begin{tikzpicture}[auto, shorten >=1pt, >=latex]
            \node[state] (e1) at (1,2) {$\id$};
            \node[state] (a11) at (-1,2) {$q_{11}$};
            \node[state] (e2) at (4,0) {$\id$};
            \node[state] (a21) at (2,0) {$q_{21}$};
            \node[state] (a22) at (0,0) {$q_{22}$};
            \node[state] (a23) at (-2,0) {$q_{23}$};
            \node[state] (a24) at (-4,0) {$q_{24}$};
            \node[state] (e3) at (4,-2) {$\id$};
            \node[state] (a31) at (2,-2) {$q_{31}$};
            \node[state] (a32) at (0,-2) {$q_{32}$};
            \node[state] (a33) at (-2,-2) {$q_{33}$};
            \node[state] (a34) at (-4,-2) {$q_{34}$};
            \node[state] (a35) at (-4,-4) {$q_{35}$};
            \node[state] (a36) at (-2,-4) {$q_{36}$};
            \node[state] (a37) at (0,-4) {$q_{37}$};
            \node[state] (a38) at (2,-4) {$q_{38}$};
            \node[state] (a39) at (4,-4) {$q_{39}$};
            \node (d) at (0,-5) {$\vdots$};
            \draw[->] (a11) edge [bend right] node [below] {$1/0$} (e1)
                      (a11) edge [bend left] node [above] {$0/1$} (e1)
                      (a24) edge node {$2/2$} (a23)
                      (a23) edge [bend right] node [below] {$1/0$} (a22)
                      (a23) edge [bend left] node {$0/1$} (a22)
                      (a22) edge node {$2/2$} (a21)
                      (a21) edge [bend right] node [below] {$1/0$} (e2)
                      (a21) edge [bend left] node [above] {$0/1$} (e2)
                      (a39) edge node {$2/2$} (a38)
                      (a38) edge node {$2/2$} (a37)
                      (a37) edge [bend right] node [above] {$1/0$} (a36)
                      (a37) edge [bend left] node {$0/1$} (a36)
                      (a36) edge node {$2/2$} (a35)
                      (a35) edge node {$2/2$} (a34)
                      (a34) edge [bend right] node [below] {$1/0$} (a33)
                      (a34) edge [bend left] node {$0/1$} (a33)
                      (a33) edge node {$2/2$} (a32)
                      (a32) edge node {$2/2$} (a31)
                      (a31) edge [bend right] node [below] {$1/0$} (e3)
                      (a31) edge [bend left] node [above] {$0/1$} (e3);
          \end{tikzpicture}
        \end{center}
        \caption{A part of the automaton form \autoref{counterex:unboundedButOnlyFiniteOrbits}. The state $\id$ appears multiple times and arrows that go from any state to the state $\id$ by fixing a letter are not drawn.}\label{figure:NonBoundedExample}
      \end{figure}

      We claim that all its orbits are finite, but that they are not uniformly bounded. We will first show the former. For this, consider some $\alpha \in \Sigma^\omega$. If $\alpha$ contains at most one letter from $\{ 0, 1 \}$, then it is easy to see that the orbit of $\alpha$ has size at most $2$ since all semigroup elements fix every occurrence of the letter $2$ in any word.
    
      Otherwise, we can write $\alpha = u a 2^{n - 1} b \beta$ for some $u \in \Sigma^*$, $a, b \in \{ 0, 1 \}$, $n > 0$ and $\beta \in \Sigma^\omega$. Notice that, from the definition, we have $q_{ij}\cdot(a 2^{n-1} b) = \id$ for all $1 \leq j \leq i^2$ with $i \neq n$. Therefore, for all $\bm{q} \in Q^+$ and $n > 0$, we have\enlargethispage{\baselineskip}
      \[
        \bm{q} \cdot a 2^{n-1} b \in \{ \id, q_{n1}, q_{n2}, \dots, q_{nn^2} \}^*
      \]
      and, thus, also
      \[
        \bm{q} \cdot u a 2^{n-1} b \in \{ \id, q_{n1}, q_{n2}, \dots, q_{nn^2} \}^* \text{.}
      \]
      Since the semigroup generated by $\mathcal{T}$ is composed of finitary automorphisms of $\Sigma^\omega$, it is locally finite (i.e. every finitely generated subsemigroup is finite). In particular, we obtain that the subsemigroup generated by $\{ a_{n1}, a_{n2}, \dots, a_{nn^2} \}$ in $\mathscr{S}(\mathcal{T})$ is finite. As we have $\bm{q} \circ \alpha = (\bm{q} \circ u a 2^{n - 1} b)[(\bm{q} \cdot u a 2^{n-1} b) \circ \beta]$, we obtain that the orbit of $\alpha$ is also finite.

      To see that the orbits are unbounded, it suffices to note that the action is transitive on all
      \[
        V_{i} = \{a_1 a_2 \dots a_{i^2} \in \Sigma^{i^2} \mid a_k \in \{0,1\} \text{ if } k\equiv 1 \bmod i, \quad a_k=2 \text{ otherwise}\}
      \]
      with $i > 0$.
    \end{counterexample}

    We have seen that the analogue of \autoref{cor:infiniteSemigroupsHaveInfiniteOrbits} -- that every infinite automaton semigroup admits a single $\omega$-word with an infinite orbit -- does not hold in the context of self-similar semigroups. Next, we will see that neither does the generalization to finitely generated subsemigroups given in \autoref{cor:infiniteSubsemigroupsHaveInfiniteOrbist}.

    For our discussion, let $\mathcal{T} = (Q, \Sigma, \delta)$ be an infinite state $S$-automaton and consider a finitely generated infinite subsemigroup of $\mathscr{S}(\mathcal{T})$. By taking the union of appropriate powers of $\mathcal{T}$, we can, without loss of generality, assume that it is generated by a finite subset $P$ of $Q$. We still know that, for every $n$, there is a word whose orbit under the action of $P^*$ has size at least $n$. However, in this case, there is not necessarily a single word with an infinite orbit.
    \begin{prop}\label{prop:counterExampleForFinitelyGeneratedSubsemigroupSelfSimilar}
      Consider the self-similar group generated by the infinite state $G$-automaton $\mathcal{T} = (Q, \Sigma, \delta)$
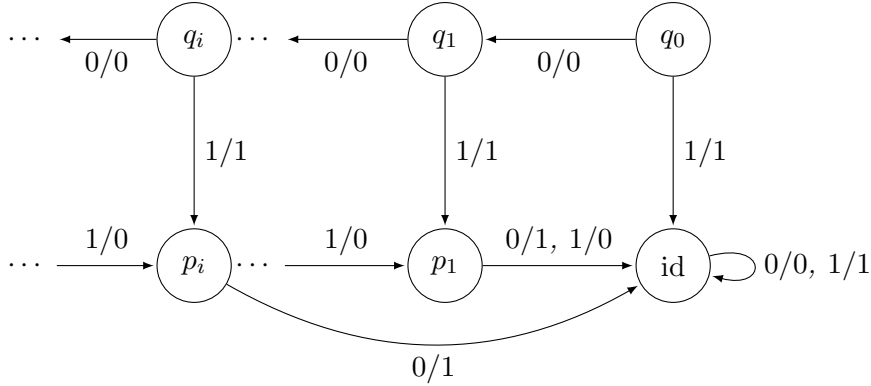
      \begin{center}
        \begin{tikzpicture}[auto, shorten >=1pt, >=latex]
          \node[state] (x0) at (3,0) {$q_0$};
          \node[state] (x1) at (0,0) {$q_1$};
          \node[state] (xn) at (-3.3,0) {$q_i$};
          \node[state] (a1) at (0,-3) {$p_1$};
          \node[state] (an) at (-3.3,-3) {$p_i$};
          \node (l1) at (-2.5,0) {$\cdots$};
          \node (l2) at (-5.5, 0) {$\cdots$};
          \node (l3) at (-2.5,-3) {$\cdots$};
          \node (l4) at (-5.5, -3) {$\cdots$};
          \node[state] (e) at (3,-3) {$\id$};
          \draw[->] (x0) edge node {$1/1$} (e)
                    (x0) edge node {$0/0$} (x1)
                    (x1) edge node {$0/0$} (l1)
                    (x1) edge node {$1/1$} (a1)
                    (xn) edge node {$0/0$} (l2)
                    (a1) edge node {$0/1$, $1/0$} (e)
                    (l3) edge node {$1/0$} (a1)
                    (an) edge [bend right] node [below] {$0/1$} (e)
                    (l4) edge node {$1/0$} (an)
                    (xn) edge node {$1/1$}(an)
                    (e) edge [loop right] node {$0/0$, $1/1$} (e);
        \end{tikzpicture}
      \end{center}
      (and its inverse). Then, both, the subsemigroup as well as the subgroup generated by $q_0$ (and possibly its inverse $q_0^{-1}$) are infinite, but all orbits $q_0^* \circ \alpha \subseteq \{ q_0, q_0^{-1} \}^* \circ \alpha$ are finite.
    \end{prop}
    \begin{proof}
      It is clear that $q_0^{\pm 1}$ fixes $0^\omega$. Thus, let $\alpha$ be an $\omega$-word different to $0^\omega$ and write it as $\alpha = 0^n 1 \beta'$ with $n \geq 0$. We further factorize $\beta' = v \beta$ for $|v| = n$. From the definition of the automaton, we see that
      \[
        q_0^{\pm 1} \circ \alpha = q_0^{\pm 1} \circ 0^n 1 v \beta = 0^n 1 (p_n^{\pm 1} \circ v) \beta
      \]
      and
      \[
        q_0^{\pm i} \circ \alpha = 0^n 1 (p_n^{\pm i} \circ v) \beta \text{.}
      \]
      Therefore, the orbit $\{ q_0, q_0^{-1} \}^* \circ \alpha$ is bounded by $| \{ 0, 1 \}|^{|v|} = 2^n$ and, thus, finite. On the other hand, we can always choose $n$ large enough so that the actions of $q_0^i$ and $q_0^j$ differ on $0^n 1 0^n$ for $i \neq j$, showing that the subsemigroup and subgroup generated by $q_0$ is infinite.
    \end{proof}
    
    It is not surprising that there are infinitely many states reachable from the finite set $P = \{ q_0 \}$ in the automaton in \autoref{prop:counterExampleForFinitelyGeneratedSubsemigroupSelfSimilar} because, otherwise, $P$ would already be the subset of some finite state $S$-automaton. This observation allows us to formulate a result similar to \autoref{cor:infiniteSubsemigroupsHaveInfiniteOrbist}: every infinite subsemigroup of a self-similar semigroup generated by a finite subset $P$ of states such that there are only finitely many states reachable from $P$ admits an $\omega$-word with an infinite orbit under the action of the subsemigroup.
  \end{section}
  
  \section*{Acknowledgement}
    The authors would like to thank the anonymous referee for the helpful remarks that stirred a discussion leading to a greatly simplified proof for the main result.

\bibliographystyle{plainurl}
\bibliography{references}

\begin{thebibliography}{10}

\bibitem{aklmp12}
Ali Akhavi, Ines Klimann, Sylvain Lombardy, Jean Mairesse, and Matthieu
  Picantin.
\newblock On the finiteness problem for automaton (semi)groups.
\newblock {\em International Journal of Algebra and Computation}, 22(06):1--26,
  2012.
\newblock \href {http://dx.doi.org/10.1142/S021819671250052X}
  {\path{doi:10.1142/S021819671250052X}}.

\bibitem{MR2643891}
Daniele D'Angeli, Alfredo Donno, Michel Matter, and Tatiana Nagnibeda.
\newblock Schreier graphs of the {B}asilica group.
\newblock {\em J. Mod. Dyn.}, 4(1):167--205, 2010.
\newblock \href {http://dx.doi.org/10.3934/jmd.2010.4.167}
  {\path{doi:10.3934/jmd.2010.4.167}}.

\bibitem{DaRo14}
Daniele D'Angeli and Emanuele Rodaro.
\newblock A geometric approach to (semi)-groups defined by automata via dual
  transducers.
\newblock {\em Geometriae Dedicata}, 174(1):375--400, 2015.
\newblock \href {http://dx.doi.org/10.1007/s10711-014-0024-x}
  {\path{doi:10.1007/s10711-014-0024-x}}.

\bibitem{DaRo16}
Daniele D'Angeli and Emanuele Rodaro.
\newblock Freeness of automaton groups vs boundary dynamics.
\newblock {\em Journal of Algebra}, 462:115--136, 2016.
\newblock \href {http://dx.doi.org/10.1016/j.jalgebra.2016.05.015}
  {\path{doi:10.1016/j.jalgebra.2016.05.015}}.

\bibitem{structurePart}
Daniele D'Angeli, Emanuele Rodaro, and Jan~Philipp W{\"a}chter.
\newblock The structure theory of partial automaton semigroups.
\newblock {\em arXiv preprint}, 2018.
\newblock \href {http://arxiv.org/abs/1811.09420} {\path{arXiv:1811.09420}}.

\bibitem{decidabilityPart}
Daniele D'Angeli, Emanuele Rodaro, and Jan~Philipp W{\"a}chter.
\newblock Automaton semigroups and groups: on the undecidability of problems
  related to freeness and finiteness.
\newblock {\em Israel Journal of Mathematics}, 2020.
\newblock \href {http://dx.doi.org/10.1007/s11856-020-1972-5}
  {\path{doi:10.1007/s11856-020-1972-5}}.

\bibitem{expandabilityPart}
Daniele D'Angeli, Emanuele Rodaro, and Jan~Philipp W{\"a}chter.
\newblock Orbit expandability of automaton semigroups and groups.
\newblock {\em Theoretical Computer Science}, 809:418 -- 429, 2020.
\newblock \href {http://dx.doi.org/10.1016/j.tcs.2019.12.037}
  {\path{doi:10.1016/j.tcs.2019.12.037}}.

\bibitem{DAngeli2019boundary}
Daniele D’Angeli, Thibault Godin, Ines Klimann, Matthieu Picantin, and
  Emanuele Rodaro.
\newblock Boundary dynamics for bireversible and for contracting automaton
  groups.
\newblock {\em International Journal of Algebra and Computation}, pages 1--19,
  2019.
\newblock \href {http://dx.doi.org/10.1142/S021819672050006X}
  {\path{doi:10.1142/S021819672050006X}}.

\bibitem{francoeur2018existence}
Dominik Francoeur and Ivan Mitrofanov.
\newblock On the existence of free subsemigroups in reversible automata
  semigroups.
\newblock {\em arXiv preprint}, 2018.
\newblock \href {http://arxiv.org/abs/1811.04679} {\path{arXiv:1811.04679}}.

\bibitem{Gilbert13}
Pierre Gillibert.
\newblock The finiteness problem for automaton semigroups is undecidable.
\newblock {\em International Journal of Algebra and Computation}, 24(01):1--9,
  2014.
\newblock \href {http://dx.doi.org/10.1142/S0218196714500015}
  {\path{doi:10.1142/S0218196714500015}}.

\bibitem{glasner2005Automata}
Yair Glasner and Shahar Mozes.
\newblock Automata and square complexes.
\newblock {\em Geometriae Dedicata}, 111(1):43--64, Mar 2005.
\newblock \href {http://dx.doi.org/10.1007/s10711-004-1815-2}
  {\path{doi:10.1007/s10711-004-1815-2}}.

\bibitem{DynSubgroup}
Rostislav~I. Grigorchuk.
\newblock Some topics of the dynamics of group actions on rooted trees.
\newblock {\em The Proceedings of the Steklov Institute of Math.}, 273:1 --
  118, 2011.

\bibitem{GriNeShu}
Rostislav~I. Grigorchuk, Volodymyr~V. Nekrashevych, and Vitaly~I. Sushchansky.
\newblock Automata, dynamical systems, and groups.
\newblock {\em Proc. Steklov Inst. Math.}, 231(4):128--203, 2000.

\bibitem{GriSa13}
Rostislav~I. Grigorchuk and Dmytro Savchuk.
\newblock Self-similar groups acting essentially freely on the boundary of the
  binary rooted tree.
\newblock In {\em Group theory, combinatorics, and computing}, volume 611 of
  {\em Contemp. Math.}, pages 9--48. Amer. Math. Soc., Providence, RI, 2014.

\bibitem{GrigSav16}
Rostislav~I. Grigorchuk and Dmytro Savchuk.
\newblock Ergodic decomposition of group actions on rooted trees.
\newblock {\em Proceedings of the Steklov Institute of Mathematics}, 2016.

\bibitem{nekrashevych2005self}
Volodymyr~V. Nekrashevych.
\newblock {\em Self-similar groups}, volume 117 of {\em Mathematical Surveys
  and Monographs}.
\newblock American Mathematical Society, Providence, RI, 2005.
\newblock \href {http://dx.doi.org/10.1090/surv/117}
  {\path{doi:10.1090/surv/117}}.

\end{thebibliography}

\end{document}